\documentclass[oribibl,11pt]{llncs}  %
\usepackage[utf8]{inputenc}
\usepackage{mathtools,amssymb,mathrsfs} %
\usepackage[normalem]{ulem}
\usepackage{float}
\usepackage{framed}
\usepackage[table,svgnames]{xcolor}
\usepackage[colorlinks=true,linkcolor=black,citecolor=MidnightBlue,urlcolor=MidnightBlue]{hyperref}
\usepackage{xspace}
\usepackage{tikz}
\usepackage{capt-of} %
\usepackage{float}
\usepackage{microtype}
\usetikzlibrary{trees}
\usepackage{arydshln}  %
\usepackage{multirow}
\usepackage{booktabs}
\usepackage{hhline}
\usepackage[linewidth=1pt]{mdframed}
\usepackage[margin=1em,font=small,labelfont=bf]{caption}

\usepackage{lmodern}
\usepackage[T1]{fontenc} %
\usepackage{comment}

\usepackage[hmargin=1in,tmargin=0.9in,bmargin=1in]{geometry}   %

\newcommand{\eps}{\varepsilon}
\newcommand{\ecc}{\mathsf{ECC}}
\newcommand{\ECC}{\ecc}

\newcommand{\model}{\mathscr{M}}

\newcommand{\silence}{\emptyset}    %
\newcommand{\erase}{\bot}

\newcommand{\T}{{\cal T}}
\newcommand{\N}{{\cal N}}
\newcommand{\hamming}{\Delta}   %
\newcommand{\dist}{\hamming}

\newcommand{\Tenc}{\mathsf{TCenc}} %
\newcommand{\Tdec}{\mathsf{TCdec}} %
\newcommand{\anc}{anc}	%
\DeclareMathOperator*{\argmin}{\arg\!\min}
\newcommand{\bc}{Blueberry code\xspace}

\newcommand{\X}{\mathcal{X}}
\newcommand{\Y}{\mathcal{Y}}
\newcommand{\Z}{\mathcal{Z}}
\newcommand{\commmain}{\comm^{\text{adp}}}

\newcommand{\maxround}{R_{\text{max}}}

\newcommand{\erate}{\mathsf{NR}}
\newcommand{\err}{\mathsf{NC}}
\newcommand{\comm}{\mathsf{CC}}
\newcommand{\RC}{{\mathsf{RC}}}

\newcommand{\ch}{\mathsf{Ch}}

\newcommand{\Mfree}{\model_\text{adp}}
\newcommand{\Mpaid}{\model_\text{term}}
\newcommand{\mmain}{\Mfree}
\newcommand{\madp}{\mmain}
\newcommand{\mterm}{\Mpaid}
\newcommand{\mna}{\text{non-adaptive}}
\newcommand{\mabort}{\mterm^{\dagger}}   %

\newcommand{\ter}{\mathsf{TER}}

\newcommand{\defn}{\overset{\Delta}{=}}

\usepackage{svn}
\usepackage{svn-multi}

\svnidlong
{$HeadURL$} {$LastChangedDate$} {$LastChangedRevision: 217 $} {$LastChangedBy$}

\begin{document}

\title{Adaptive Protocols for Interactive~Communication}

\author{Shweta Agrawal\inst{1} \and Ran Gelles\inst{2} \and Amit Sahai\inst{3}}
\institute{
I.I.T Delhi, \email{shweta@cse.iitd.ac.in} \and
Princeton University,  \email{rgelles@cs.princeton.edu} \and
University of California, Los Angeles, \email{sahai@cs.ucla.edu}
}

\maketitle

\thispagestyle{plain}
\pagestyle{plain}

\begin{abstract}
How much adversarial noise can protocols for interactive communication tolerate? This question was examined by Braverman and Rao (\textit{IEEE Trans.\@ Inf.\@ Theory}, 2014) for the case of ``robust'' protocols, where each party sends messages only in fixed and predetermined rounds. 
We consider a new class of \emph{non-robust} protocols for Interactive Communication, which we call {\it adaptive} protocols. Such protocols adapt \emph{structurally} to the noise induced by the channel in the sense that both the order of speaking, and the length of the protocol may vary depending on observed noise.\looseness=-1

We define models that capture adaptive protocols and study upper and lower bounds on the permissible noise rate in these models. When the length of the protocol may adaptively change according to the noise, we demonstrate a protocol that tolerates noise rates up to~$1/3$.  When the order of speaking may adaptively change as well, we demonstrate a protocol that tolerates noise rates up to~$2/3$. 
Hence, adaptivity circumvents an impossibility result of~$1/4$ on the fraction of 
tolerable noise (Braverman and Rao, 2014). 
\end{abstract}

\thispagestyle{empty}
\setcounter{page}{0}
\newpage

\section{Introduction}
\label{sec:intro}

One of the fundamental questions considered by Computer Science is ``What is the best way to encode information in order to recover from channel noise''? This question was studied most notably by Shannon, in a pioneering work \cite{shannon48} which laid the foundation of the rich area of information theory. Shannon considered this question in the context of one way communication, where one party wants to transmit a message ``once and for all'' to another. More recently, in a series of beautiful papers, Schulman \cite{schulman92,schulman93,schulman96} generalized this question to subsume \textit{interactive communication}, i.e.\@ the scenario where 
two remote parties perform some distributed computation by 
``conversing'' with each another in an interactive manner, so that each subsequent message depends on all messages exchanged thus far. Surprisingly, Schulman showed that, analogous to the case of one way communication, it is indeed possible to embed any interactive protocol~$\pi$ within a larger protocol~$\pi'$ so that~$\pi'$ computes the same function as~$\pi$ but additionally provides the requisite error correction to tolerate noise introduced by the channel.

The noise in the channel may be stochastic, in which error occurs with some probability, or adversarial, in which the channel may be viewed as a malicious party Eve who disrupts communication by injecting errors in the worst possible way. In this work we focus on adversarial noise.  
Schulman~\cite{schulman93,schulman96} provided a construction that turns a protocol~$\pi$ with communication complexity~$T$, to a noise-resilient~$\pi'$ which communicates at most~$O(T)$ symbols, and can recover from an adversarial (bit) noise rate of at most~$1/240$. 
This result was later improved by Braverman and Rao~\cite{BR11,BR14}, who provided a protocol that can recover from a (symbol) noise rate up to~$1/4 -\eps$ and also communicates at most $O(T)$~symbols.  

Both the above constructions assume  \emph{robust protocols}.
Intuitively speaking, robust protocols are synchronized protocols that have a fixed length and a predetermined ``order of communication''. In this class of protocols, each party knows at every time step whose turn it is to speak and whether the protocol has terminated, since these properties are fixed in advance and \emph{independent} of the noise introduced by the adversary.
However, one can imagine more powerful, general protocols where 
the end point of the protocol or the order of speaking are not predetermined but rather depend
on the observed transcript, that is, on the observed noise.
While Braverman and Rao show that for any robust protocol $1/4$~is an upper bound on the tolerable noise rate, they explicitly leave open the question of whether non robust protocols
admit a larger amount of noise.

We address this question by considering two types of non-robust protocols, that
allow for greater adaptivity in the behavior of the participants. 
First, we allow the length of the protocol to be adaptively specified during
the protocol by its participants.
Next, we consider even greater adaptivity 
and allow the party that speaks next in the protocol to be adaptively chosen by
the participants of the protocol. 
In both these situations, we show that increasing adaptivity allows for a
dramatic increase in the noise resilience of protocols.

We draw attention of the reader to the fact that while for robust protocols, Yao's~\cite{Yao79} model is almost universally accepted as natural and meaningful, it is far less obvious what is the right way to model non-robust protocols, or even if there is a unique choice. Defining models to capture adaptivity is subtle, and several choices must be made, for example in how adversarial noise is budgeted and in how to model rounds in which there is no consensus regarding who the speaker is. Different modeling choices lead to different protocol capabilities and we believe it is important to explore the domian of this very young area in order to find settings that are both natural and admit protocols with higher noise resilience.

In a recent work, Ghaffari, Haeupler, and Sudan~\cite{GHS14,GH14}, proposed one natural set of choices to model adaptivity, and provided efficient protocols in that model which resist noise rates of up to~$2/7$, surpassing the maximal resilience of the non-robust case. In this paper we make a different, but arguably just as natural, set of choices, which lead to adaptive protocols with even higher noise resilience. We proceed to summarize the most salient differences in our modeling choices and the ones of~\cite{GHS14}.

First, the model in~\cite{GHS14} does not permit adaptive modification of the length of the protocol, while our model does. 
To the best of our knowledge, our work is the first to consider 
varying length interactive protocols and their noise resilience.
The second main difference is that in~\cite{GHS14} the channel may be used to communicate only in one direction at each round.
Specifically in~\cite{GHS14}, at each round, each party decides either to only \emph{talk} or to only \emph{listen}: if both parties talk at the same round, a collision occurs and no symbol is transferred, and if both listen at the same round, they receive some adversarial symbol not counted towards the adversary's budget. 
In our model, on the other hand, both parties may talk at the same round without causing any collision (similar to the case of robust protocols~\cite{schulman96,BR14}). The adaptivity stems from the parties' ability to individually choose at each round, whether they talk or not. 

The two different modeling choices taken by~\cite{GHS14} and by us lead to different bounds on the noise an adaptive protocol can handle. For instance, while the protocols of~\cite{GHS14} can handle up to a relative noise of~$2/7$, our protocols can resists a higher noise rate of~$1/3$ if the length of the protocol may adaptively change, or noise rate of up to~$2/3$ when both the length and the order of speaking adaptively vary. 
We now give more details about our adaptive model and the noise rate our protocols can resist.

\subsection{Our Results: Adaptive Length}
We begin by considering adaptive protocols 
in which the \emph{length} of the protocol may vary as a function of the noise,
however the order of speaking is still predetermined. 
Specifically, each party individually decides whether to continue participating in the protocol, or terminate and give an output. We denote the class of such protocols as~$\mterm$ (see formal definition in Section~\ref{sec:mterm-model}).

Intuitively, changing the length of the protocol is useful for two reasons. 
First, the parties may realize that they still did not complete the computation, and communicate more information in order to complete the task. 
On the other hand, the parties may see that the noise level is so high that there is no hope to correctly complete the protocol. 
In this case the parties should abort the computation, since for such a high noise level, 
the protocol is not required to be correct anyway. 
The difficult part for the parties is, however, to be able to distinguish between the first case and the second one in a coordinated way and despite the adversarial noise.

If the length of the protocol is not fixed (and subsequently, its communication complexity), the noise rate must be defined with care. Generalizing the case of fixed-length protocols, we consider the ratio of corrupted symbols out of all the symbols that were communicated \emph{in that instance}, and call this quantity the  \emph{relative} noise rate. We emphasize that both the numerator and the denominator of this ratio vary in adaptive protocols.

Our main result for this type of adaptivity  is
a protocol that resists relative noise rates of up to~$1/3$ (Theorem~\ref{thm:protocol-third}). 
The protocol works in two steps: in the first step Alice communicates her input to Bob using some standard error correction code; 
in the second step Bob estimates the noise that occurred during the first step, and then  he communicates his input to Alice using an error correction code with parameters that depend on his noise estimation. In general, the more noise Bob sees during the first step, the less redundant his reply to Alice would be---if there was a lot of noise during the first part, the adversary has less budget for the second part, and the code Bob uses can be weaker.

The communication complexity of our protocol above is a constant factor (where the constant 
depends on the channel quality) times the
input lengths of Alice and Bob.
However, our protocol requires the parties to communicate their inputs, even in cases
where the lengths of the inputs may be very long with respect to the communication complexity of the best noiseless protocol; thus, the \emph{rate} of this coding strategy (the length of the noiseless protocol divided by the length of the resilient protocol) can be vanishing when the length of the noiseless protocol tends to infinity.  
Nevertheless, our coding protocol serves as an important proof of concept for the strength of this model:
Indeed, in the non-adaptive setting, the rate of the coding scheme has no effect
on the noise resilience. E.g., an upper bound of~$1/4$ holds for coding schemes even when their rate is vanishing~\cite{BR14}.  

In addition to our schemes, we show an upper (impossibility) bound of $1/2$ on the tolerable noise in that model (Theorem~\ref{thm:mterm-half}). 
We emphasize that previous impossibility proofs (i.e., \cite{BR14}) crucially use the property of robustness: 
in robust protocols there always exists a party that speaks at most half of the symbols, whose identity is known in advance, making it a convenient target for adversarial attack. 
Contrarily, in adaptive protocols the party that speaks less may depend on the noise and vary throughout the protocol. 
We provide a new impossibility bound 
by devising an attack that corrupts \emph{both} parties with rate~$1/2$, and carefully arguing that at least one of the parties must terminate before it learns the correct output.

\begin{table}[htb]
\vspace{-1em}
\begin{center}
\small
\begin{tabular}{llll@{}}
\toprule
\textbf{Model}		 & \textbf{Lower Bound $\alpha$}$\quad$ 	&	\textbf{Upper Bound $\beta$}$\quad$	& \textbf{Ref.} \\
\midrule
$(\mna)$ 	&	$1/4$	&	$1/4$		& \cite{BR14}			\\
$\mterm$  \phantom{xxxxxxxx}		&	$1/3$	&	$1/2$		& \S\ref{sec:mterm}	\\
\bottomrule
\end{tabular}
\end{center}
\caption{Summary of our bounds for the $\mterm$ model, compared to the non-adaptive model.
$\alpha$ and $\beta$ are the lower (existence) and upper (impossibility) bounds on the allowed noise rate: for any function there exists a protocol that withstands noise rate $c$ if $c<\alpha$. Yet, there exists a function for which  no protocol withstands noise rate~$\beta$. 
}
\vspace{-3.5em}
\label{tab:res-mterm}
\end{table}

\vspace{-1.5ex}
\subsection{Our Results: Adaptive Order of Speaking}
Next, we define the $\mmain$~model in which we allow the order of speaking to depend on the noise (see formal definition in Section~\ref{sec:mmain-model}). Specifically, at each round each party decides whether it sends the next symbol or it keeps silent; the other side, respectively, either learns the symbol that was sent, or receives ``silence''.\footnote{A similar notion of party keeping silent was used in interactive protocols over \emph{noiseless} channels, by~\cite{DFO10,IW10}.}  
We stress that silent rounds, i.e.~when no message is delivered, are not counted towards the communication, or otherwise the model becomes equivalent to the~$\mterm$ model.
We note that this type of adaptivity also implies a varying length of the protocol, e.g., in order to terminate, a party simply keeps silent and disregards any incoming communication.

Similar to the $\mterm$ model, 
the adversary is allowed to corrupt any transmission, and
we measure the noise rate as the ratio of corrupted transmissions to the communicated (non-silent) symbols.
It is important to emphasize that the adversary is not limited to only corrupting symbols, but it can also \emph{create} a symbol when a party decides to keep silent, or \emph{remove} a transmitted symbol leading the receiving side to believe the other side is silent. This makes a much stronger adversary\footnote{It is easy to show that unless we give the adversary the power to insert and delete symbols, the model is too strong and the question of resisting noise becomes trivial: in that case the protocol can encode a `0' as a silent transmission, and a `1' as a non-silent transmission, thus perfectly resisting any possible noise.} that may induce relative noise rates that exceed~1.\looseness=-1

Here 
we construct an adaptive protocol which crucially uses both the ability to remain silent as well as the ability to vary the length of the protocol, to withstand noise rates~$<2/3$ (Theorem~\ref{thm:ProtocolTwoThirds}). 
The protocol behaves quite similar to the $\mterm$ protocol that achieves noise rates up to~$1/3$ with an additional layer of encoding  that takes advantage of being able to remain silent, and provides another factor of~$2$ in resisting noise. We name this new layer of code \emph{silence encoding}. The idea behind this layer is that using $k$ non-silent transmissions, one can obtain a code with distance~$2k$.  Then, in order to cause a decoding error, the adversary must invest $\tfrac12 2k +1$ corruptions, or otherwise either the correct symbol can be decoded, or the symbol becomes an erasures (which is easier to handle than an error). Note that for the special case of~$k=1$, \emph{two} corruptions are required to cause decoding of an incorrect symbol (or otherwise, the adversary only causes an erasure).

The main drawback of the above protocol is that its length (i.e., its round complexity) may be very large with respect to the length of the optimal noiseless protocol; thus it has a vanishing rate. 
Next, we restrict the discussion only to adaptive protocols 
whose length is linear in the length of the optimal noiseless protocol (thus their rate is a positive constant and not vanishing), 
and show a protocol with non-vanishing rate that tolerates noise rates of up to~$1/2$ (Theorem~\ref{thm:ProtocolHalf}). The protocol is based on the optimal (non-robust) protocol of~\cite{BR14} with an additional layer of of silence encoding which effectively forces the adversary to ``pay twice'' for each error it wishes to make. This way the protocol can withstand twice the number of errors than~\cite{BR14}.

\begin{table}[htb]
\vspace{-1em}
\begin{center}
\small
\begin{tabular}{lccl@{}}
\toprule
\textbf{Model}		 & \textbf{Noise Resilience} $\quad$	& \textbf{Non-Vanishing Rate} $\quad$	& \textbf{Ref.} \\
\midrule
$\mna$ 						&	$1/4$	& $\surd$	&	\cite{BR14}					\\
\midrule
$\mmain$						&	$2/3$	& 		&	\S\ref{sec:dsa}, \S\ref{app:freepos}				\\
$\mmain$ 				 		&	$1/2$	& $\surd$ 	&	\S\ref{app:mainCRprotocol} \\ %
$\mmain$ (shared randomness)	&	$1$		& $\surd$	&	\S\ref{app:protocolSharedRand}	\\
\midrule
$\mmain$ over erasure channels	&	$1$		& $\surd$ 	&	 \S\ref{sec:dsa} (\S\ref{app:protocolSharedRand}) \\

\bottomrule
\end{tabular}
\end{center}
\caption{Summary of the noise resilience of our protocols in the $\Mfree$ model. 
For any function~$f$, and for any constant $c$ less than the resilience, 
there exists a protocol that correctly computes~$f$ over any channel with relative noise rate~$c$.
Note that $1$ is a trivial impossibility bound for the $\mmain$~model, as the adversary can delete the entire communication.
}
\vspace{-2.5em}
\label{tab:res-madp}
\end{table}

If we relax the model to permit the parties to share some randomness unknown to the adversary, then we can construct a protocol that withstands an optimal $1 - \eps$ fraction of errors (Theorem~\ref{thm:main-shared}) and also achieves non-vanishing rate.
The key technique here is to adaptively repeat transmissions that were corrupted by the adversary: each symbol is sent multiple times until the other side indicates that the symbol was received correctly. However, now the adversary can corrupt this ``feedback'' and falsely indicate that a symbol was received correctly by the other side. To prevent such an attack we use the shared randomness to add a layer of error-detection (via the so called Blueberry code~\cite{FGOS15}). The adversary, without knowing the randomness, has a small probability to corrupt a symbol so it passes the error-detection layer, and  corrupts the sensitive ``feedback'' symbols with only a negligible probability. 

An interesting observation is that we can apply our methods to the setting of \emph{erasure channels} and obtain a protocol with linear round complexity (i.e., with a non-vanishing rate)
and erasure resilience of $1-\eps$ without the need for a shared randomness (Corollary~\ref{cor:erasure}).
We note that for non-adaptive protocols over erasure channels, $1/2$ is a tight bound on the noise: a noise of $1/2-\eps$ is achievable via the Braverman-Rao protocol (see~\cite{FGOS15}) or via the simple protocol of Efremenko, Gelles and Haeupler~\cite{EGH15}; on the other hand, a noise rate of $1/2$ is enough to erase the entire communication of a single party, thus disallowing any interaction~\cite{FGOS15}. Our protocol for adaptive settings hints that adaptivity can double the resilience to noise (similar to the effect of possessing preshared private randomness~\cite{FGOS15}, etc.).
Our bounds for the $\madp$ model are summarized in Table~\ref{tab:res-madp}.\looseness=-1

\subsection{Related Work.} 
As mentioned in the introduction, the study of coding for interactive communication was initiated by Schulman~\cite{schulman92,schulman93,schulman96} who provided protocols for interactive communication using {\it tree codes} (see Appendix~\ref{app:mainCRprotocol} for a definition and related works). 
In his work, Schulman considered both the stochastic as well as adversarial noise model, and for the latter provided a protocol that resists (bit) noise rate~$1/240$. Braverman and Rao~\cite{BR11,BR14} improved this bound to~$1/4$ by constructing a different tree-code based protocol (which is efficient except for the generation of tree codes). 
Braverman and Efremenko~\cite{BE14} considered the case where $\alpha$ fraction of the symbols from Alice to Bob are corrupted and $\beta$ fraction of the symbols in the other direction are corrupted. For any point $(\alpha,\beta)\in [0,1]$ they determine whether or not interactive communication (with non-vanishing rate) is possible. 
This gives a complete characterization of the noise bounds for the non-adaptive case.

Over the last years, there has been great interest in interactive protocols, considering various properties of such protocols
such as their efficiency~\cite{GMS11,GMS14} (stochastic noise),~\cite{BK12,BN13,GH14,BKN14} (adversarial noise), their noise resilience under different assumptions and models~\cite{FGOS13,BNTTU14,EGH15,FGOS15},
their information rate~\cite{KR13,Pankratov13,Haeupler14,GH15} and other properties, such as privacy~\cite{CPT13,GSW14} or list-decoding~\cite{GHS14,GH14,BE14}. 
We stress that all the works prior to this work (and to the independent work~\cite{GHS14,GH14}), assume the robust, non-adaptive setting.

The only other work that studies adaptive protocols is the abovementioned work of Ghaffari, Haeupler, and Sudan~\cite{GHS14}, which makes different modeling decision than our work.  
Ghaffari et al.\@ show that in their adaptive model, $2/7$~is a tight bound on fraction of permissible noise. 
The length of the protocol obtained in~\cite{GHS14} is quadratic in the length of the noiseless protocol, thus its rate is vanishing. However Ghaffari and Haeupler~\cite{GH14} later improve the length to be linear while still tolerating the optimal $2/7$~noise of that model. Allowing the parties to preshare randomness increases the admissible noise to~$2/3$.
We stress again that the setting of~\cite{GHS14} and ours are incomparable. Indeed, the tight $2/7$ bound of~\cite{GHS14} does not hold in our model and we can resist relative noise rates of up to $1/3$ or $2/3$ in the $\mterm$ and $\mmain$ models respectively. Similarly, while $2/3$ is the bound on noise when parties are allowed to share randomness in~\cite{GHS14}, in our model, the relative noise resilience for this setting is~$1$.

We note that interactive communication can also be extended to the multiparty case, following the more simple two party case, see e.g.~\cite{RS94,JKL15,HS14,ABEGH15}. The adaptive setting is particularly relevant
to asynchronous multiparty settings (as in~\cite{JKL15}) which is closely related to the  $\mmain$ model we present here.

Interactive (noiseless) communication in a model where parties are allowed to remain silent (similar to the case of the $\mmain$ model), was introduced by Dhulipala, Fragouli, and Orlitsky~\cite{DFO10}, who consider the communication complexity of computing symmetric functions in the multiparty setting.  In their general setting, each symbol $\sigma$ in the channel's alphabet has some weight $w_\sigma \in [0,1]$ and the weighted communication complexity, both in the average and worst case, is analyzed for a specific class of functions. Remaining silent can be thought of sending a special ``silence'' symbol, whose weight is usually~0.
Impagliazzo and Williams~\cite{IW10} also consider communication complexity given a special silence symbol for the two-party case. They establish a tradeoff between the communication complexity and the round complexity. Additionally, they relate these two measures to the ``standard'' communication complexity, i.e., without using a silence symbol.

\section{Protocols with an Adaptive Length}
\label{sec:mterm}

In this section, we study the $\mterm$ model in which parties adaptively determine the
length of the protocol by (locally) terminating at will. First, let us formally define the model.

\subsection{The $\mterm$ model}
\label{sec:mterm-model}
We assume Alice and Bob wish to compute some function $f: \X\times\Y \to \Z$ where Alice holds some input $x\in \X$ and Bob holds $y\in \Y$. The sets $\X,\Y$ and $\Z$ are assumed to be of finite size. 
We assume Alice and Bob run a protocol $\pi=(\pi_A,\pi_B)$, over a channel controlled by a malicious Eve.
At every step of the protocol, $\pi$ defines a message over some alphabet~$\Sigma$ of finite size (which may depend only on the targeted noise resilience) to be transmitted by each party as a function of the party's input, and the received messages so far.

In this model, each party sends symbols according to a predetermined order.
Let $I_A,I_B \subseteq \mathbb{N}$ be the round indices in which Alice and Bob talk, respectively. Note that $I_A$ and $I_B$ may be overlapping but we may assume without loss of generality that there are no ``gaps'' in the protocol, i.e. $I_A \cup I_B = \mathbb{N}$. The channel expects an input from party $P\in \{A,B\}$ only during rounds in~$I_P$. 
The behavior of Alice in the protocol is as follows (Bob's behavior is symmetric):

\begin{itemize}
\item In a given round $i \in I_A$, if Alice has not terminated, 
she transmits a symbol $a_i \in \Sigma \cup \{\silence\}$, where $\Sigma$ is the channel's alphabet and $\silence$ is a special symbol we call \emph{silence}.
\item
At the beginning of any round $i\in \mathbb{N}$, Alice may decide to terminate. In that case
she outputs some value, %
 sets $\ter_A = i$ and stops participating in the protocol. This is an irreversible decision. 
\item In every round $i \in I_A$ where $i \ge \ter_A$,  Alice's input to the channel is the special silence symbol~$a_i = \silence$. 
\item Eve may corrupt any symbol,  including silence, transmitted by either party. Thus, she acts upon transmitted symbols via the function $\ch: \Sigma \cup \{\silence\} \rightarrow \Sigma \cup \{\silence\}$, conditioned on the parties input, Eve's random coins and the transcript so far. Note that even after Alice has terminated, $a_i=\silence$ is sent over the channel and still might be corrupted by Eve.
\end{itemize}

Next we formally define some important measures of a protocol.
For a specific instance of the protocol we define the \emph{Noise Pattern} 
$E\in (\Sigma\cup \{\bot\})^*$ incurred in that instance in the following way.
Assume that Alice sends $(a_1,a_2,\ldots)$ 
and Bob sends $(b_1,b_2,\ldots)$, then 
$E=((e_{a_1},e_{a_2}, \ldots), (e_{b_1},e_{b_2}, \ldots))$  
so that $e_{a_i}=\bot$ if $\ch(a_i)=a_i$ and otherwise, $e_{a_i}=\ch(a_i)$, and similarly for~$e_{b_i}$. 

\begin{definition} 
For any protocol $\pi$ in the $\mterm$ model, 
we define the following measures for any given instance of $\pi$ running on inputs~$(x,y)$ 
suffering the noise pattern~$E$:
\begin{enumerate}
\item Communication Complexity: 
\(
\comm^{\text{term}}_\pi(x,y,E) \defn \left\lvert[\ter_A-1] \cap I_A\right\rvert + \left\lvert[\ter_B-1] \cap I_B\right\rvert,
\)
where $[n]$ is defined as the set $\{1,2, \ldots, n\}$.

\item Round Complexity:
\(
\RC ^{\text{term}}_\pi(x,y,E) \defn \max(\ter_A, \ter_B).
\) 
\item Noise Complexity: 
\(
\err^{\text{term}}_\pi(x,y,E) \defn \\
 \big|\{ i\in I_A \mid i  < \RC^{\text{term}}_\pi \; ,\;  \ch(a_i)\neq a_i\} \big| + 
  \big|\{ i\in I_B \mid i < \RC^{\text{term}}_\pi \; ,\;  \ch(b_i)\neq b_i\} \big|.
\)
\item Relative Noise Rate: 
\( \erate^{\text{term}}_\pi(x,y,E)  \defn {\err^{\text{term}}_\pi (x,y,E)}/{\comm^{\text{term}}_\pi (x,y,E)}.
\)
\end{enumerate}
\end{definition}
In order to avoid protocols that never halt, we assume there exists a global constant $\maxround$ and that for any input and noise pattern  $\RC  \le \maxround$.
Finally, we say that a protocol is \emph{correct} if both parties output~$f(x,y)$. We say that a protocol \emph{resists $\eps$-fraction of noise} (or, resists noise rate~$\eps$), if the protocol is correct (on any input) whenever the relative noise rate induced by the adversary is at most~$\eps$. 
Note that if a protocol resists noise rate of~$\eps$, and the relative noise in a specific instance is higher than~$\eps$, there is no guarantee on the output of the parties.

\subsection{Tolerating noise rates up to~$1/3$}
\label{sec:mterm-achieve}
In this section we show how to use the power of adaptive termination in order to circumvent the $1/4$ bound on the noise of~\cite{BR14}.
Below, we provide a protocol that resists noise rate~$1/3-\eps$ in the $\mterm$ model.

\begin{theorem}\label{thm:protocol-third}
For any function $f$ and any $\eps>0$, there exist a protocol $\pi$ for in the $\mterm$ model, 
that resist a noise rate of~$1/3-\eps$.
\end{theorem}

\begin{proof}
We assume parties' inputs are in~$\{0,1\}^n$. 
We will use a family of good error correcting codes $\ECC_i: \{0,1\}^n \to \Sigma^{c_in}$ with
$i=1,\ldots, i_{\max}$. Each such code corrects up to $1/2-\eps$ fraction of errors 
while having a constant rate $1/c_i$ and using a constant alphabet $\Sigma$, both of which depend on~$\eps$.
The redundancy of each code increases with $i$, i.e., $c_{i+1} > c_{i}$. 
Moreover, these codes will have the property that for any $x$, $\ECC_i(x)$ is a prefix of $\ECC_{j}(x)$ for any $j>i$. This can easily be done with random linear codes, e.g., by randomly choosing a large generating matrix of size $n\times c_{i_{\max}}n$ and encoding $\ECC_i$ by using a truncated matrix using only the first $c_i$~columns.

Formally, for any~$n$ and $\eps>0$, let $\{\ECC_i\}$ be a family of error correcting codes as described above 
and let $j$ be such that  $ c_j\cdot 4\eps \ge c_1$. Set $I_A =\{1, \ldots,  c_jn \}$ and
$I_B =\{ c_jn+1,\   c_jn +2, \ldots\}$.
\begin{enumerate}
\item Alice encodes her input using $\ECC_j$, and sends the codeword over to Bob in the first $c_jn$ rounds of the protocol.
\item After $c_jn$ rounds, Bob decodes Alice's transmission to obtain~$\tilde x$. Let $t$ be the Hamming distance between the codeword Bob receives and $\ECC_j(\tilde x)$.
\item Bob continues in an adaptive manner:
\begin{enumerate}
	\item if $t <(1/2-\eps)c_jn$ Bob encodes his input using a code $\ECC: \{0,1\}^n \to \{0,1\}^{2c_jn-4t}$. Note that the maximal value $t$ can get is $(1/2-\eps)c_jn$ which makes $2c_jn-4t > 4\eps c_jn \ge c_1n$, so a suitable code can always be found.
	\item\label{step:bobabort} otherwise, Bob aborts. 
	\end{enumerate} 	
\item After completing his transmission, Bob terminates and outputs~$f(\tilde x,y)$.
\item Alice waits until round $3c_jn$ and then decodes Bob's transmission to obtain $\tilde y$ and outputs $f(x,\tilde y)$.
\end{enumerate}

Suppose an instance of the protocol that is not correct, and let us
analyze the noise rate in that given failed instance.
First, note that if Bob aborts at step~\ref{step:bobabort}, the noise rate is clearly larger than $1/3$.
Next, assume Bob decodes a wrong value $\tilde x\ne x$. Note, that the minimal distance of the code is $1-2\eps$, thus given that Bob measures Hamming distance~$t$,  Eve must have made at least $(1-2\eps)c_jn-t$ corruptions. The total communication in this scenario is
$c_jn+2(c_jn-2t)$ which yields a relative noise rate $\frac{(1-2\eps)c_jn-t}{3c_jn-4t}$, with a minimum of~$1/3-O(\eps)$.

On the other hand, if Bob decodes the correct value~$\tilde x = x$, and measures Hamming distance~$t$,  Eve must have made $t$~corruptions at Alice's side. To corrupt Bob's codeword, she must perform at least~$({1/2-\eps})(2c_jn-4t)$ additional corruptions, yielding a relative noise rate at least $\frac{t+ (1/2-\eps)(2c_jn-4t)}{3c_jn-4t}= \frac{(1-2\eps)c_jn-(1-4\eps) t}{3c_jn-4t}$ which also obtains a minimal value of~$1/3-O(\eps)$.

There is still a remaining subtlety of how Alice knows the right code to decode. Surely, if there is no noise, Bob's transmission is delimited by silence. However, if Eve turns  the last few symbols transmitted by Bob into silence, she might cause Alice to decode with the wrong parameters.
This is where we need the prefix property of the code, which keeps a truncated codeword a valid encoding of Bob's input for smaller parameters. 
Eve has no advantage in shortening the codeword: 
if Eve tries to shorten a codeword of $\ECC_i$ into $\ECC_j$ with $j<i$ and then corrupt the shorter codeword,  she will have to corrupt
$(c_i-c_j + (1/2-\eps)c_j)n \ge (1/2-\eps)c_in$ symbols, which only increases her noise rate.
Similarly, if she tries to enlarge $\ECC_i$ into $\ECC_j$ with $j>i$, in order to cause Alice to decode the longer codeword incorrectly, Eve will have to perform at least $(1/2-\eps)c_j$ corruptions which is again more than needed to corrupt the original message sent by Bob. 
\qed
\end{proof}

\subsection{Impossibility bound}

Next, we show that in the general $\mterm$ case, no interactive protocol resists a noise rate of $1/2$ or more. At a high level, the attack proceeds by changing $1/2$ of {\it both} Alice and Bob's messages so that whoever terminates first is completely confused about their partner's input. One must exercise some care to ensure that the attack is well defined, but this high level idea can be formalized, as shown below.
\begin{theorem}
\label{thm:mterm-half}
There exists a function $f$, such that any adaptive protocol $\pi$ 
for $f$ in the $\mterm$ setting, cannot resist noise rate of~$1/2$.
\end{theorem}

\begin{proof}
Assume $f$ is the identity function on input space $\{0,1\}^n\times \{0,1\}^n$, and consider an adaptive protocol~$\pi$ that computes $f$. We show an attack that causes a relative noise rate of at most $1/2$ and causes at least one of the parties to output the wrong value.

Fix two distinct inputs $(x,y)$ and $(x',y')$. 
Given any input $\xi$ out of the set $\{ (x,y), (x',y), (x,y'), (x',y')\}$ we can define an attack on $\pi(\xi)$. The attack will change \emph{both} parties' transmissions in the following way:
Alice's messages will be changed to the ``middle point'' between what she should send given that her input is~$x$ and what she should send given that her input is~$x'$ (i.e., to a string which has the same Hamming distance from what Alice sends on~$x$ and on~$x'$). At the same time, Bob's messages are changed to the middle point between what he should send given that his input is~$y$ and what he should send given that his input is~$y'$.

Specifically, at each time step, Eve considers the next transmission of Alice on $x$ and on $x'$ (given the transcript so far). If Alice sends the same symbol in both cases, Eve doesn't do anything.  Otherwise, Eve alternates between sending a symbol from Alice's transcript on $x$ and on~$x'$. Note that the attack is well defined even if Alice has already terminated on input $x$ but not on $x'$,\footnote{Recall that once Alice terminates, we assume the symbol $\emptyset$ is being transmitted by the channel.} although we only use the attack until Alice aborts on one of the inputs.  
Corrupting Bob's transmissions is done in a similar way.

Next, consider the termination time of the attack on inputs $\{ (x,y), (x',y), (x,y'), (x',y')\}$. There exists an input whose termination time is minimal. Denote this input by~$\xi^*$ and assume, without loss of generality, that Alice is the party that terminates first when the attack is employed on~$\pi(\xi^*)$.
It follows that when employing the above attack on any of the other three inputs in $\{ (x,y), (x',y), (x,y'), (x',y')\}\setminus\{\xi^*\}$, the termination time of the parties are not smaller than Alice's termination time in the instance $\pi(\xi^*)$ under the same attack. Without loss of generality, assume $\xi^*=(x,y)$.

Finally note that when Alice terminates, she cannot tell whether Bob holds $y$ or $y'$. Indeed, up to the point she terminates, the attack on $\xi^*=(x,y)$ and the attack on $(x,y')$ look exactly the same from Alice's point of view. This is because Bob does not terminate before Alice (for \emph{both his inputs}!), and our attack changes Bob's messages in both instances in a similar way. Thus Alice's view is identical for both Bob's inputs, and she must be wrong at least on one of them.
Note that such an attack causes at most $1/2$ noise in each direction up to the point where Alice terminates (there's no need to continue in the attack after that point). Thus, the total corruption rate is at most~$1/2$.
\qed
\end{proof}

One subtlety that arises from the above proof, is the ability of a party to convey some amount of information by the specific time it terminates. In order to better understand the power of termination yet without allowing the parties to convey information solely by their time of termination, we define the $\mabort$ model,
which is exactly the same as $\mterm$ defined in Section~\ref{sec:mterm-model} above, 
except that if $\min\{\ter_A,\ter_B\} \ne \maxround$ then the parties’ output is defined as an invalid output~$\bot$.

In Appendix~\ref{app:mqwh} we analyze protocols in the $\mabort$ model which are fully utilized: at every round both parties send a single symbol over the channel. We show that $1/4$~is a tight bound on the noise in that case. While the protocol of~\cite{BR14} is enough to resist such a noise level (even without using the adaptivity), 
an impossibility bound of noise~$\ge1/4$ is not implied by previous work. In the appendix we prove the following,
\begin{theorem}\label{thm:impAbort}
There exists a function $f:\{0,1\}^n\times\{0,1\}^n\to \{0,1\}^{2n}$ such that
for any fully utilized adaptive protocol $\pi$ for~$f$ in the $\mabort$ model,
$\pi$ does not resist a noise rate of~$1/4$.
\end{theorem}

\section{Protocols with an Adaptive Order of Communication}\label{sec:dsa}
In this section we extend the power of protocols to adaptively determine the order of speaking as a function of the observed transcript and noise. To this end, 
at every round each party decides whether to send an additional symbol, or to remain silent.
We begin by defining the~$\mmain$ model.

\subsection{The $\mmain$ model}
\label{sec:mmain-model}
Similar to the $\mterm$ model, we assume Alice and Bob wish to compute some function $f: \X\times\Y \to \Z$ where Alice holds some input $x\in \X$ and Bob holds $y\in \Y$. The sets $\X,\Y$ and $\Z$ are assumed to be of finite size. 
We assume a channel with a finite alphabet $\Sigma$ that can be used by either of the parties at any round. 
Parties in this model behave as follows (described for Alice, Bob's behavior is symmetric):

\begin{itemize}
\item In a given round $i$, Alice decides whether to speak or remain silent. If Alice speaks, she sends a message~$a_i \in \Sigma$; if Alice is silent,~$a_i=\silence$. 

\item Eve may corrupt any symbol,  including the silence symbol, transmitted by either party. Thus,  Eve acts upon transmitted symbols via the function
$\ch: \Sigma \cup \{\silence\} \rightarrow \Sigma \cup \{\silence\}$,
conditioned on the parties' input, Eve's random coins and the transcript so far. 
\item The corresponding symbol received by Bob is~$\tilde{a}_i=\ch(a_i)$.
\item We assume the protocol terminates after a finite time. There exists a number $\maxround$ at which both parties terminate and output a value as a function of their input and the communication.
\end{itemize}\vspace{-1ex}
For a specific instance of the protocol we denote the messages sent by the parties
$M=(a_1,b_1,a_2,b_2,\ldots)$ in that instance, and the \emph{Noise Pattern} 
$E=(e_{a_1},e_{b_1},\ldots)$ so that $e_{a_i}=\bot$ if $\ch(a_i)=a_i$ and otherwise, $e_{a_i}=\ch(a_i)$, and similarly for~$e_{b_i}$. We will treat $E$ and $M$ as strings of length $2R_{\max}$ and refer to their $i$-th character as $E_i$ and $M_i$.

\begin{definition} 
For any protocol $\pi$ in the $\mmain$ model, and for any specific instance of the protocol on inputs $(x,y)$ with noise pattern $E$ we define: 
\begin{enumerate}
\item Communication Complexity:
\(
\comm^{\text{adp}}_\pi(x,y,E) \defn |\{ i \le 2\maxround \mid M_i \not= \silence \}|,
\)
where $M$ is the message string observed when running $\pi$ on inputs $(x,y)$ with noise pattern~$E$.

\item Noise Complexity:
\(
\err^{\text{adp}}_\pi(x,y,E) \defn |\{ i \le 2\maxround \mid E_i \ne \bot \}|.
\)

\item Relative Noise Rate: 
\(
\erate^{\text{adp}}_\pi(x,y,E) \defn {\err^{\text{adp}}_\pi(x,y,E) }/{\comm^{\text{adp}}_\pi(x,y,E) }.
\) 
\end{enumerate}
\end{definition}
As before, the protocol is correct if both parties output~$f(x,y)$. The protocol is said to resist $\eps$-fraction of noise (or, a noise rate of~$\eps$) if the protocol is correct (on any input) whenever the relative noise rate is at most~$\eps$. Note that the relative noise rate may exceed~$1$.

\subsection{Resilient Protocols in the $\mmain$ model}
In this section, we study several protocols in the $\mmain$ model 
that achieve better noise resilience than in the robust case. 
The main result of this section is a protocol that tolerates relative noise rates of up to~$2/3$.
\begin{theorem}\label{thm:ProtocolTwoThirds}
Let $\X,\Y,\Z$ be some finite sets.
For any function $f : \X\times \Y \to \Z$ there exists an adaptive protocol~$\pi$ for $f$ in the $\mmain$ model that resists noise rates below~$2/3$.
\end{theorem}
This protocol builds upon the protocol constructed in Theorem \ref{thm:protocol-third} but additionally adds a layer of coding that takes advantage of the partially-utilizing nature of message delivery in this model, which we call \emph{silence encoding}. More formally,
\begin{definition}\label{def:SE1}
Let $X=\{x_1, x_2, ..., x_n\}$ be some finite, totally-ordered set.
The \emph{silence encoding} is a code $\textit{SE}_1: X\to (\Sigma\cup\{\silence\})^n $ that encodes $x_i$ into a string $y_1,...,y_n$ where 
$\forall j\ne i$,   $y_j=\silence$ and $y_i\ne \silence$. \\
\end{definition}
\vspace{-1em}
Intuitively, such an encoding has the property that \emph{two} transmissions must be corrupted in order to make the receiver decode an incorrect message, while only a single symbol is transmitted. This, along with the technique that tolerates relative noise rates of up to~$1/3$ when only the length of the protocol is adaptive, yields the claimed result. See Appendix~\ref{app:freepos} for full details and proof of Theorem~\ref{thm:ProtocolTwoThirds}.

\smallskip
While the protocol of Theorem~\ref{thm:ProtocolTwoThirds} obtains noise rate resilience of~$2/3$ and very small communication complexity, it has double exponential round complexity with respect to the round complexity of the best noiseless protocol.
Our next theorem limits the round complexity to be linear, thus yielding a coding scheme with a non-vanishing rate. Specifically, it shows that for any $\eps>0$ we can emulate any protocol $\pi$ of length~$T$ (defined in the noiseless model) 
by a protocol $\Pi$ in the $\mmain$ model, which takes at most $O(T)$~rounds and resists noise rate of~$1/2-\eps$.

\begin{theorem}\label{thm:ProtocolHalf}
For any constant~$\eps >0$ and for any function~$f$, there exists an interactive protocol in the $\mmain$ model 
with round complexity $O(\comm_f)$, 
that resists a relative noise rate of~$1/2-\eps$.
\end{theorem}
The protocol follows the emulation technique set forth by Braverman and Rao~\cite{BR14}, and requires a generalized analysis for channels with errors and erasures as performed in~\cite{FGOS15} albeit for a completely different setting. The key insight is that silence encoding forces the adversary to pay twice for making an error (or otherwise to cause ``only'' an erasure). This allows doubling the maximal noise rate the protocol resists.
The  proof appears in Appendix~\ref{app:mainCRprotocol}.

\smallskip
Finally, we extend the model by allowing the parties to share a random string, unknown to the adversary.
We show that shared randomness setup 
allows the relative noise rate to go as high as~$1-\eps$.
Formally, 
\begin{theorem}\label{thm:main-shared}
For any  small enough constants $\eps >0$ and for any function~$f$, 
there exists an interactive protocol in the $\mmain$~model 
with round complexity $O(\comm_f)$ such that,
if the adversarial relative corruption rate is at most~$1-\eps$,
the protocol correctly computes~$f$ with overwhelming success probability over the choice of the shared random string.
\end{theorem} 

The proof appears in Appendix~\ref{app:protocolSharedRand}. At a high level, the main idea is to adaptively repeat transmissions that were corrupted by the adversary. This turns each transmission into a varying-length message whose length (i.e., the number of repetitions) is determined by the relative noise at that message. This forces Eve to spend more and more of her budget in order to corrupt a single transmission, since she needs to corrupt all the repetitions that appear in a single message. The shared randomness is used as a means of detecting corruptions (similar to~\cite{FGOS15}), converting most of Eve's noise into  easy to handle \emph{erasures}.
Each detected corruption is replaced with an erasure mark and treated accordingly.  
It is immediate then, that the same resilience of~$1-\eps$ holds for protocols over the \emph{erasure channel}, even when no preshared randomness is available: 
such channels can only make ``erasures'' to begin with,
so there is no need for preshared randomness in order to detect corruptions.
\begin{corollary}\label{cor:erasure}
For any  small enough constant~$\eps >0$ and for any function~$f$, there exists an interactive protocol in the $\mmain$~model over an erasure channel
that has round complexity $O(\comm_f)$ and that correctly computes~$f$ as long as 
the adversarial relative erasure rate is at most~$1-\eps$.
\end{corollary}

\goodbreak

\section*{Acknowledgements}
We would like to thank Hemanta Maji and Klim Efremenko for useful discussions.

Research supported in part from a DARPA/ONR PROCEED award, NSF Frontier Award 1413955, NSF grants
1228984, 1136174, 1118096, and 1065276, a Xerox
Faculty Research Award, a Google Faculty Research Award, an equipment
grant from Intel, and an Okawa Foundation Research Grant. This
material is based upon work supported by the Defense Advanced Research
Projects Agency through the U.S. Office of Naval Research under
Contract N00014-11-1-0389. The views expressed are those of the author
and do not reflect the official
policy or position of the Department of Defense, the National Science
Foundation, or the U.S. Government.

\bibliographystyle{alpha}
\bibliography{coding-short}

\appendix
\section*{Appendix}

\section{The $\mabort$ model: impossibility bounds}
\label{app:mqwh}
In this section we study upper (impossibility) bounds on the admissible noise in the $\mabort$ model. 
We consider \textit{fully utilized} protocols in which both parties send a symbol at every round (i.e., $I_A=I_B=\mathbb{N}$). In this setting we show an impossibility bound of~$1/4$ on the amount of tolerable noise, matching the achievable resilience of protocols in this setting~\cite{BR14}.
Specifically, we provide an adversarial strategy that always wins with error rate~$<1/4$. Note that Braverman and Rao \cite{BR14} showed a similar result for non-adaptive protocols. Informally speaking, their proof goes along the following lines: Eve picks the player, say Bob, who speaks for fewer slots, and changes half his messages so that the first half corresponds to input~$y$ while the second half corresponds to~$y'$. Now, Eve's noise rate is at most~$1/4$, and Alice cannot tell whether Bob's input is $y$ or~$y'$ and cannot output the correct value.

The above strategy does not carry over to the $\mabort$ model. Specifically, the above attack is not \emph{well defined}. Indeed, Eve can inject messages in the first half of the attack, by running Bob's part of~$\pi$ on the input~$y$. However, when Eve wishes to switch to~$y'$, she now needs to run $\pi$ on input $y'$ \emph{given the transcript so far}, say, given $\mathsf{tr}(y)$. It is possible that $\pi(\cdot,y')$ conditioned on $\mathsf{tr}(y)$ is not defined, for example upon occurrence of $\mathsf{tr}(y)$ given input $y$, Bob may have  already terminated and Eve cannot conduct the second part of the attack.

We address this issue by demonstrating a more sophisticated attack
that does not abruptly switch $y$ to $y'$ after half the messages, but rather gradually moves from $y$ towards~$y'$. 
That is, at any time during the protocol the adversary's relative noise rate is at most~$1/4$, 
therefore the parties' ability to prematurely terminate doesn't give them any power%
.\footnote{In fact, the bound we obtain is $1/4-O(1/k)$ where $k$~is the round complexity of~$\pi$. Therefore, the only hope to obtain protocols that resist any noise rate strictly less than~$1/4$ is having infinite protocols. This is however beyond the scope of this work, and is left as an open question.} Recall the statement of Theorem~\ref{thm:impAbort},

\medskip
\noindent
\textbf{Theorem~\ref{thm:impAbort}.}
\textit{
There exists a function $f:\{0,1\}^n\times\{0,1\}^n\to \{0,1\}^{2n}$ such that
for any adaptive protocol $\pi$ for~$f$ in the fully utilized $\mabort$~model,
$\pi$ does not resist a noise rate of~$1/4$.}
\medskip

Before we prove the theorem we show the following technical lemma, which is the main idea of our proof.
Denote the Hamming distance of two strings by $\dist(\cdot,\cdot)$. In order to cause ambiguity when decoding a codeword from $\{x,y\}$,  one needs to corrupt at most $(\dist(x,y)+1)/2$ symbols, and this can be done in a ``rolling'' manner. Formally,
\begin{lemma}\label{lem:rollingChange}
Assume $\mathbb{F}$ is some finite field. 
For any two strings $x,y\in \mathbb{F}^n$ %
there exists a string~$z\in \mathbb{F}^n$ such that
\[
\dist(x+z,x) \ge \dist (x+z,y)
\]
and for any $j\le n$, $w(z_1,\ldots,z_j) \le \frac{j+1}2$,
where $w(\cdot)$ is the Hamming weight function.
\end{lemma}
\begin{proof}
We begin by proving that when $\dist(x,y)$ is even, a more restricted form of the lemma holds, 
namely, that for any $j\le n$, $w(z_1,\ldots,z_j) \le \frac{j}2$.
We prove this by induction on the hamming distance $d=\dist(x,y)$.
The case of $d=2$ is easily obtained by setting $z$ to be all zero except for the second index where $x$ and $y$ differ. Now assume the hypothesis holds for an even $d$ and consider $d+2$.
Split $x=x_1x_2$ and $y=y_1y_2$ such that $|x_1|=|y_1|$ and $\dist(x_1,y_1)=d$ (thus $\dist(x_2,y_2)=2$). Let $u,v$ be the strings guaranteed by the induction hypothesis for $x_1,y_1$ and $x_2,y_2$ respectively, and set $z=uv$.

By the way we construct $z$, it holds
\(
\dist(x+z,x) \ge \dist (x+z,y).
\)
Moreover, for any $j<|x_1|$ we know that 
$w(z_1,\ldots,z_j) \le \frac{j}2$, by the induction hypothesis. Note that $w(v)$ is at most $1$, and that $v_1=0$ by the construction of the base case.
Then it is clear that the claim holds for $j=|x_1|+1$; 
for any $j>|x_1|+1$ we get
$w(z_1,\ldots,z_j) = w(u)+w(v_1,...,v_{j-|x_1|+1}) \le \frac{|x_1|}{2} + 1  \le \frac{j}{2}$.

Completing the proof of the original lemma (where $d$ can be odd and the weight is $\le \frac{j+1}2$) 
is immediate. If~$d$ is odd we construct~$z$ by using the induction lemma (of the even case) over the prefix  with hamming distance $d-1$ and change at most a single additional index, which is located \emph{after} that prefix.
Assume that the prefix is of length $n_{\rm{prefix}}$. 
The claim holds for any $j \le n_{\rm{prefix}}$ due to the induction hypothesis.
For any $j> n_{\rm{prefix}}$ it holds that 
\[
w(z_1,\ldots,z_{j})= w(z_1,\ldots,z_{n_{\rm{prefix}}})+w(z_{n_{\rm{prefix}}+1},\ldots, z_j) \le \frac{n_{\rm{prefix}}}2+1 \le \frac{j+1}{2}.
\]
\qed
\end{proof}

We now continue to proving that $1/4$ is an upper bound of the permissible noise rate.
\begin{proof}(\textbf{Theorem~\ref{thm:impAbort}.})
Let $f$ be such that for any $y,y'$, $f(x,y)\ne f(x,y')$, for instance, the identity function $f(x,y)=(x,y)$, and let~$\pi$ be any adaptive protocol for~$f$.
Consider the transcripts of $\pi$ up to round 10.\footnote{10 is obviously arbitrary, and has the sole purpose of avoiding the edge case in which Eve corrupts the first couple of rounds, possibly causing (a relative) noise rate higher than~$1/4$.}
By the pigeon-hole principle for large enough~$n$, there must be $y,y'$ that for some $x$ produce the same transcript up to round~10. 
Let  $m = \min \left\{ TER_A(x,y), TER_A(x,y')\right\}$.

The basic idea is the following.
Assuming no noise, 
let $t$ be Bob's messages in $\pi(x,y)$ up to round $m$ 
and $t'$ be Bob's messages in $\pi(x,y')$ up to round $m$.
Using Lemma~\ref{lem:rollingChange} 
Eve can change $t$ into $t+z$ (starting from round 10), so that  $\dist(t+z,t)\ge \dist(t+z,t')$ and Eve's relative noise rate never exceeds~$1/4$. 
Furthermore, Eve can change $t'$ into $t'+z'=t+z$ and also in this case Eve's relative noise rate  never exceeds~$1/4$: the string $z'_{11},...,z'_{m}$ must satisfy, for any index $10< j \le m$, that $w(z'_{11},\ldots,z'_{j})\le \frac{j+1}2$ (this follows from the way we construct $z$ and the fact that $\dist(t'+z',t)=\dist(t+z,t)\ge \dist(t+z,t')=\dist(t'+z',t')$). Thus, the relative noise rate made by Eve up to round~$j$  is at most $\frac{(j-10+1)/2}{2j} <1/4$. 
The same argument should be repeated until we reach the bound on the round complexity $\ter_\pi$, which we formally prove in Lemma~\ref{lem:adaptiveBob}, yet before getting to that we should more carefully examine the actions of both parties during this attack.

Consider Alice actions when the messages she receives are $t+z=t'+z'$.
She can either (i) abort (output $\bot$), (ii) output $f(x,y)$ or (iii) output $f(x,y')$, however, her actions are independent of Bob's input (since her view is independent of Bob's input). Assuming Eve indeed never goes beyond $1/4$, it is clear that Eve always wins in case (i). For case (ii) Eve wins on input $(x,y')$ and for case (iii), Eve wins on input $(x,y)$.

However, while in the above analysis Alice's actions must be the same between the two cases of $t\to t+z$ and $t'\to t'+z'$, 
this is not the case for Bob. 
We must be more careful and consider Bob's possible adaptive reaction to  errors made by Eve. In other words, Bob, noticing Alice's replies, may either abort, or send totally different messages so that his transcript is neither $t$ nor~$t'$. We now show that even in this case Eve has a way to construct $z,z'$ and never exceed a relative noise rate of~$1/4$.

\begin{lemma}\label{lem:adaptiveBob}
Assume $\pi$ takes $k$ rounds. 
Eve always has a way to change (only) messages sent by Bob, so that
Alice's view is identical between an instance of~$\pi(x,y)$ and of~$\pi(x,y')$, while Eve
corrupts no message up to round~$10$ 
and at most $(k-9)/2$ messages between rounds $11$ and~$k$ (incl.).
\end{lemma}
\begin{proof}
We prove by induction. The base case where $k\le 10$ is trivial. 

Assume the lemma holds for some even $k$, and we prove for $k+1$ and $k+2$. 
By the induction hypothesis, Eve can cause the run of $\pi(x,y)$ and $\pi(x,y')$ look identical in Alice's eyes while corrupting at most $(k-9)/2$~messages after round~$10$.

Let $t$ denote the next two messages (rounds $k+1,k+2$) 
sent by Bob in the instance of $\pi(x,y)$ and $t'$ in the instance of $\pi(x,y')$.\footnote{Note that $t,t'$ are conditioned on the noise Eve has introduced throughout round~$k$.} 
There are strings $z,z'$ such that
$w(z),w(z')\le 1$ and $t+z=t'+z'$. Assume we construct $z$ via the the construction of Lemma~\ref{lem:rollingChange} then also $z_1=0$ and $z'_2=0$. Also note that $t,t'$ are independent of errors made in rounds $k+1, k+2$ (Bob `sees' that his message at $k+1$ has been changed at round~$k+2$ at the earliest, thus this information can affect only his messages at rounds $>k+2$).

At round $k+1$ the amount of corrupt messages (in both cases) is at most
\[
\left\lfloor \frac{k-9}2\right\rfloor + 1 \underset{k\text{ is even}}{=} \frac{k-10}{2}+1
\le \frac{(k+1) -9}{2}
\]
And the same holds for round $k+2$ (for both cases).
\qed
\end{proof}

With the above lemma, Eve can always cause Alice to be confused between an instance of $\pi(x,y)$ and $\pi(x,y')$ by inducing, at any point of the protocol, a relative noise rate of at most 
\[
\frac{\frac{k-9}2}{2k} < \frac14.
\]
Therefore, unless one of the parties aborts\footnote{As before, Eve needs not corrupt any message after one of the parties aborts, since she is always within her budget.}, Alice outputs a wrong output. In all these cases the protocol is incorrect while the noise rate is at most~$1/4$.
\qed
\end{proof}

 Attempts to extend the above proof to work for the fully utilized $\mterm$~model runs into a hurdle created by 
parties' ability to communicate information about their inputs by the time of aborting.
Indeed, in the above attack Alice learns Bob's inputs (since they were never corrupted), and Bob might be able to distinguish 
$x$ from~$x'$ by whether or not Alice has prematurely aborted 
(i.e., according to the number of silence symbol implicitly communicated by the channel after Alice terminates).

\section{Proof of Theorem~\ref{thm:ProtocolTwoThirds}}
\label{app:freepos}
We now show that every function can be computed by an $\mmain$ protocol that can suffer noise rates less than~$2/3$.
The main technique used in this section is a simple code that takes advantage of the `silence' symbols, which we call \emph{silence encoding} defined in Definition~\ref{def:SE1} for a simple special case, and below for the general case:
\begin{definition}
Let $X=\{x_1, x_2, ..., x_n\}$ be some finite, totally-ordered set.
The \emph{$k$-silence encoding} is a code $\textit{SE}_k: X\to (\Sigma\cup\{\silence\})^{kn}$
that encodes $x_i$ into a string $y_1,...,y_{kn}$
where all $y_j=\silence$ except for the $k$ indices $y_{(i-1)k+1}, \ldots, y_{ik}\in \Sigma^k$.
\end{definition}
Decoding a $k$-silence-encoded codeword is straightforward. The receiver tries to find a message~$x_i$ whose encoding minimizes the Hamming distance to the received codeword. If the string that minimizes the distance is not unique, the decoder marks this event as an \emph{erasure} and outputs~$\bot$. The event where the decoder decodes $x_j\ne x_i$ is called an \emph{error}.
Both encoding and decoding can be done efficiently.

We note the following interesting property of $k$-silence encoding: in order to cause ambiguity  in the decoding (i.e., an erasure), the adversary must change at least $k$ indices in the codeword. 
Moreover, in order to make the decoder output an incorrect value (i.e., an error), the adversary must make 
at least $k+1$ changes to the codeword. Specifically for $k=1$, a single corruption always causes an \emph{erasure} (i.e., ambiguity), while in order to make a decoding \emph{error}, at least 2 transmissions must have been changed.

\smallskip
We are now ready to prove our main theorem of this section. 

\smallskip
\noindent\textbf{Theorem~\ref{thm:ProtocolTwoThirds}.}
\textit{
Let $\X,\Y,\Z$ be some finite sets.
For any function $f : \X\times \Y \to \Z$ there exists an adaptive protocol~$\pi$ for $f$ in the $\mmain$ model that resists noise rates below~$2/3$.}
\begin{proof}
The protocol is composed of two parts, similar to the protocol of Theorem~\ref{thm:protocol-third}: in the first part Alice communicates her input to Bob and in the second part Bob communicates his input to Alice. After the first part, Bob estimates the error injected and proceeds to the second part only if the noise-rate is low enough to correctly complete the protocol, or is high enough so that the adversary will surely exceed its budget by the time the protocol ends (as these two cases are indistinguishable). In addition, Bob's message crucially depends  on the amount of error Eve introduced in the channel.

Assume the channel is defined over some alphabet~$\Sigma$ and denote one of the alphabet's symbols by~`$\sigma$'.
For any $k\in \mathbb{N}$ define $\pi$ on inputs $x_i,y_j \in \X \times \Y$ in the following way:
\begin{enumerate}
\item Alice communicates a $k$-silence encoding of her input, namely,
she waits~$k\cdot (i-1)$ rounds and then sends the symbol~$\sigma$ for $k$ consecutive rounds.
\item Bob waits until round $k|\X|$ and decodes the codeword sent by Alice. Bob adaptively chooses his actions according to the following cases:
\begin{enumerate}
\item if there is ambiguity regarding what $x_i$~is, Bob aborts.
\item otherwise, Bob decodes some~$x_{i'}$.
Let $t$ be the difference between the number of $\sigma$ symbols Bob received during those rounds that ``belong'' to $x_{i'}$ and the number of $\sigma$'s received during the rounds that ``belong'' to a value $x_{i''}$, where $x_{i''}$ is the  2nd best decoding of the received codeword (when decoding by minimizing Hamming distance)

Bob communicates  his input~$y_j$ using the following $2t$-silence encoding: he waits $2k\cdot(j-1)$ rounds and then sends the symbol~$\sigma$ for $2t$~consecutive rounds. \\ Then, Bob outputs $f(x_{i'},y_j)$ and terminates.
\end{enumerate}
\item Alice waits until round $R_{\max} \triangleq k|\X|+2k|\Y|$, and decodes the codeword sent by Bob. If there is ambiguity regarding the value of $y_j$, Alice aborts. Otherwise, she obtains some $y_{j'}$. Alice then outputs $f(x_i, y_{j'})$ and terminates.
\end{enumerate}

Let us analyze what happens at round~$k|\X|$. As mentioned above, in order to cause ambiguity at that round, Eve must change at least $k$~transmissions. In this case Bob aborts at round~$k|\X|$; observe that neither of the  parties communicates any symbol after round $k|\X|$, thus their total communication for this instance is $k$~symbols. This implies noise rate of at least~$1$.

If, on the other hand, at round $k|\X|$ there was no ambiguity, one of two things must have happened:
either Bob correctly decodes Alice's input, or he decodes a wrong input. 
First assume the latter, which implies that at least $k+1$  corruptions were done. 
Since there is no ambiguity, we know that $t>0$ and it must hold that Eve made $e\ge k+t$ corruptions.
Then, by the end of the protocol, the relative noise rate is at least~$\frac{e}{k+2(e-k)}$. This value decreases as $e$ increases, up till the point where $e=2k$ at which it gets a minimal value of $2/3$. 
Eve has no incentive to perform more than $e=2k$ corruptions, this will only increase the relative noise rate without changing the actions of Bob.

Now assume Bob decodes the correct value, thus Eve wins only if Alice decodes a wrong value from Bob or aborts. We consider two cases. 
(i) If Eve has corrupted~$e<k$ symbols by round~$k|\X|$, then 
Bob will send his input via $2t$-silence encoding, where $t\ge k-e$. Thus, in order for Alice to decode a wrong value (or abort), Eve must perform at least additional $2t$ corruptions, yielding a relative noise rate of at least $\frac{e+2t}{k+2t}$. 
Under the constraints that $0 \le e \le k-1$ and $k-e\le t \le k$, it is easy to verify that
$$\frac{e+2t}{k+2t} \ge 1- \frac{t}{k+2t} \ge \frac23.$$
(ii) If Eve has made $e\ge k$ corruptions by round~$k|\X|$, yet Bob decoded the correct value, 
Eve will have to corrupt additional $2t$ symbols to to cause confusion at Alice's side. 
This implies a relative noise rate of at least
\[
\frac{e+2t}{k+2t} \ge \frac{k+2t}{k+2t} =1.
\]
\qed
\end{proof}

\section{Proof of Theorem~\ref{thm:ProtocolHalf}}
\label{app:mainCRprotocol}

In this section we prove Theorem~\ref{thm:ProtocolHalf}, and show a protocol with resilience $1/2-\eps$ and non-vanishing rate in the $\mmain$~Model.
 First, let us recall some primitives and notations that will be used in our proof.
We denote the set $\{1,2,\ldots,n\}$ by~$[n]$, and
for a finite set $\Sigma$ we denote by~$\Sigma^{\le n}$ the set~$\cup_{k=1}^n \Sigma^k$.
The Hamming distance $\dist(x,y)$ of two strings $x,y \in \Sigma^n$
is the number of indices~$i$ for which~$x_i \ne y_i$, and the Hamming weight of some string,
is its distance from the all-zero string, 
$w(x) = \dist(x,0^n)$.
Unless otherwise written, $\log()$~denotes the binary  logarithm (base~2).

A $d$-ary \emph{tree-code}~\cite{schulman96} over alphabet~$\Sigma$ is a rooted $d$-regular tree of arbitrary depth~$N$ whose edges are labeled with elements of~$\Sigma$.
For any string $x\in [d]^{\le N}$, a $d$-ary tree-code~$\T$ implies
\emph{an encoding}
of~$x$,
$\Tenc(x)=w_1w_2..w_{|x|}$ with $w_i\in \Sigma$, defined by concatenating the labels along the path defined by $x$, i.e., the path that begins at the root and whose  $i$-th node is the $x_i$-th child of the $(i\!-\!1)$-st node.

For any two paths (strings) $x,y\in [d]^{\le N}$ of the same length~$n$,
let~$\ell$ be the longest common prefix of both $x$~and~$y$. Denote by $\anc(x,y)= n-|\ell|$ the distance from the $n$-th level to the least common ancestor of paths $x$~and~$y$.
A tree code  has distance $\alpha$ if for any $k\in [N]$ and any distinct
$x,y\in [d]^{k}$, the Hamming distance of $\Tenc(x)$ and $\Tenc(y)$ is at least $\alpha \cdot \anc(x,y)$.

For a string $w \in \Sigma^n$, decoding $w$ using the tree code $\T$ means returning the string $x\in [d]^n$ whose encoding  minimizes the Hamming distance to the received word, namely,
$$\Tdec(w) = \argmin_{x\in [d]^n} \Delta( \Tenc(x), w)\text{.}$$

A theorem by Schulman~\cite{schulman96} proves that for any $d$ and $\alpha<1$ there exists a $d$-ary tree code of unbounded depth and distance $\alpha$ over alphabet of size $d^{O(1/(1-\alpha))}$.
However, no efficient construction of such a tree is yet known.
For a given depth~$N$, Peczarski~\cite{peczarski06} gives a randomized construction for a tree code with $\alpha=1/2$ that succeeds with probability at least~$1-\epsilon$, and requires alphabet of size at least~$d^{O(\sqrt{\log\epsilon^{-1}})}$.
Braverman~\cite{braverman12} gives a sub-exponential (in~$N$) construction of a tree-code, and 
Gelles, Moitra and Sahai~\cite{GMS11,GMS14} provide an efficient construction of a randomized  relaxation of a tree-code of depth~$N$, namely a \emph{potent tree code}, which is powerful enough as a substitute for a tree code in most applications. 
Finally, Moore and Schulman~\cite{MS14} suggested an efficient construction which is based on a conjecture on some exponential sums.

\bigskip
We now prove Theorem~\ref{thm:ProtocolHalf}.
For any function $f$ and any constant $\eps>0$, we construct a protocol that correctly computes $f$ as long as the relative noise rate does not exceeds $1/2-\eps$.

Let $\eps>0$ be fixed, and let $\pi$ be an interactive protocol in the noiseless model for $f$, in which the parties exchange bits with each other for up to $T$ rounds. 
We begin by turning $\pi$ into a resilient version $\pi_{BR}$ which resist noise rate of up to $1/4-\eps$, using techniques from~\cite{BR14}.
The protocol takes $N=O(T)$ rounds in each of which both parties send
 a  message  over some finite alphabet~$\Sigma$
\begin{lemma}[\cite{BR14}]
For every $\eps$ there is an alphabet $\Sigma$ of size $O_\eps(1)$ 
such that any binary protocol $\pi$ can be compiled to a protocol $\pi_{BR}$ of  $N=O_\eps(|\pi|)$ rounds in each of which both parties send a symbol from $\Sigma$. For any input $x,y$, both parties output $\pi(x,y)$ if the fraction of errors is at most $1/4-\eps$.
\end{lemma}
\noindent The conversion is described in~\cite{BR14}; We give more details about this construction in the proof of Lemma~\ref{lem:Nhigh}.

Next, we construct a protocol $\Pi$ that withstands noise rate of $1/2-\eps$.
The parties run $\pi_{BR}$, yet each symbol from $\Sigma$ is silence-encoded. That is, every round of $\pi_{BR}$ in which a party sends some symbol~$a\in \Sigma$ 
is expanded into $|\Sigma|$~rounds of~$\Pi$ in which a single symbol `$\sigma$' is sent at a timing  that corresponds to the index of $a$ in the total ordering of $\Sigma$. The channel alphabet used in~$\Pi$ is thus unary. Decoding is performed by minimizing  Hamming distance and the decoder obtains either a symbol of $\Sigma$ or an erasure mark~$\erase$.  

From this point and on, we regard only rounds of $\pi_{BR}$ protocol, ignoring the fact that each such `round' is composed of $|\Sigma|$ mini-rounds.
Denote by $\N(i,j)$ the `effective'  noise-rate between rounds $i$ and $j$, for which an erasure is counted as a single error and decoding the wrong symbol of $\Sigma$ is counted as two errors. Formally,
assume that at time~$n$, Alice sends a symbol $a_n\in \Sigma$, 
and Bob receives $\tilde{a}_n\in \Sigma \cup \{\bot\}$, possibly with added noise or an erasure mark
(similarly, Bob sends $b_n \in \Gamma$, and Alice receives~$\tilde{b}_n$).
\begin{definition}\label{def:N}
Let the effective noise in Alice's transmissions be
\[
{\cal N}_A(i,j) = | \{ k\mid  i\le k \le j, \tilde{a}_k=\bot \}| + 2| \{ k \mid i \le k \le j, \tilde{a}_k \notin \{ a_k, \bot\} \}| \text{,}
\]
and similarly define ${\cal N}_B(i,j)$ for the effective noise in Bob's transmissions. The \emph{effective number of corruptions} in the interval $[i,j]$ is ${\cal N}(i,j) = {\cal N}_A(i,j) + {\cal N}_B(i,j) $.
\end{definition}
The following lemma states that if the $\pi_{BR}$ fails, then~$\N$ must be high.

\begin{lemma}[\cite{FGOS15}]\label{lem:Nhigh}
Let $\eps>0$ be fixed and let $|\pi_{BR}|=N$. %
If $\pi_{BR}$ fails, then 
\[
\N(1,N) \ge (1-\eps)^2 N.
\]
\end{lemma}
With this lemma, the proof of the theorem is immediate:
recall that with silence encoding, causing an erasure costs at least one corruption and causing an error costs at least two corruptions. Observe that 
$\commmain_{\Pi}=\commmain_{\pi_{BR}}=2N$, then 
if the amount of corruptions is limited to $1/2-\eps$,
\[
\max\  {\cal N}(1,N) = (1-2\eps)N \;\; < (1-\eps)^2N
\]
where the maximum is over all possible noise-patterns of at most $(1/2-\eps)\cdot 2N$ corruptions.

Finally, we give the proof for Lemma~\ref{lem:Nhigh}. Parts of this analysis were taken as-is from~\cite{FGOS15} and we re-iterate them here (with the authors' kind permission) for self containment.

\begin{proof}(Lemma~\ref{lem:Nhigh}.)
Let us recall how to construct a constant (non-vanishing) rate protocol~$\pi_{BR}$ for computing $f(x,y)$ over a noisy channel out of an interactive protocol $\pi$ for the same task that assumes a noiseless channel~\cite{BR14}. We assume that~$\pi$ consists of $T$~rounds in which Alice and Bob send a single bit according to their input and previous transmissions. Without loss on generality, we assume that Alice sends her bits at odd rounds while Bob transmits at even rounds. We can view the computation of~$\pi$ as a root-leaf walk along a binary tree in which odd levels correspond to Alice's messages and even levels to Bob's, see Figure~\ref{fig:pitree}.
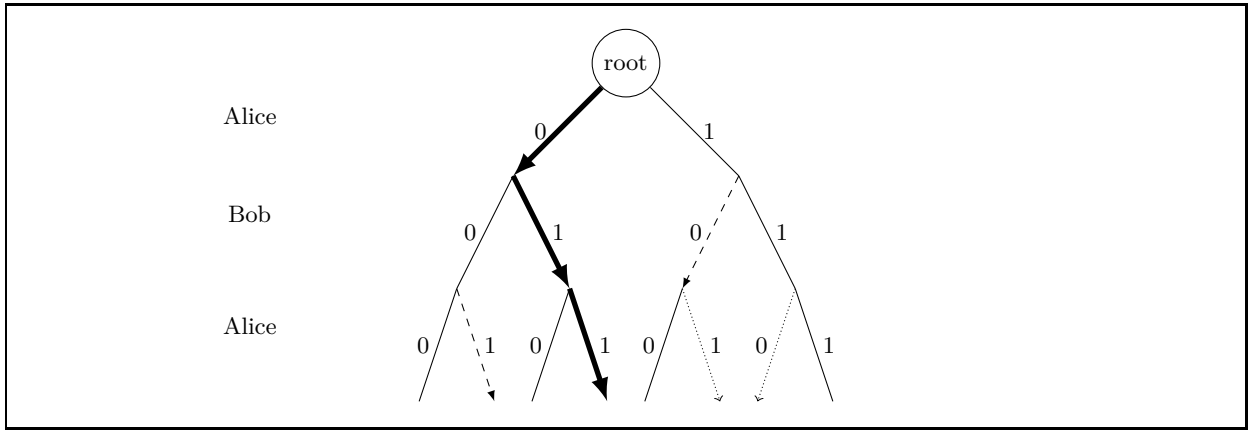
\begin{figure}[htb]
\center
\begin{framed}
\begin{tikzpicture}[scale=1,
	level/.style={sibling distance=30mm/#1},
	tik/.style={line width=2pt,-latex},
	dash/.style={dashed,-latex},
	reg/.style={thin,solid}]
\node [circle,draw] (z){root}
	child[thin] {
			child{
				child{
				edge from parent
				node [left] {0}
				}
				child{
				edge from parent [dash]
				node [right] {1}
				}			
			edge from parent [reg]
			node [left] {0}
			}
			child{
				child{
				edge from parent [reg]
				node [left] {0}
				}
				child{
				edge from parent [tik]
				node [right] {1}
				}			
			edge from parent [tik]
			node [right] {1}
			}
	edge from parent [tik]
	node [left] {0}
	}
	child {
			child{
				child{
				edge from parent [reg]
				node [left] {0}
				}
				child{
				edge from parent [densely dotted,->]
				node [right] {1}
				}			
			edge from parent [dash]
			node [left] {0}
			}
			child{
				child{
				edge from parent [densely dotted,->]
				node [left] {0}
				}
				child{
				edge from parent
				node [right] {1}
				}			
			edge from parent
			node [right] {1}
			}
	edge from parent
	node [right] {1}
	}
	;
\node at (-5,-0.7) {Alice};
\node at (-5,-2) {Bob};
\node at (-5,-3.5) {Alice};
\node at (5,0) {\phantom{Alice}};  %
\end{tikzpicture}
\end{framed}
\caption{\footnotesize A $\pi$-tree showing the path $P$ (bold edges) taken by Alice and Bob for computing $f(x,y)$. Dashed edges represent the hypothetical reply of Alice and Bob given that
a different path $P'$ was taken (when such replies are defined).}
\label{fig:pitree}
\end{figure}

In order to obtain a protocol that withstands (a low rate of) channel noise, Alice and Bob \emph{simulate} the construction of path~$P$ along the $\pi$-tree. The users transmit edges of~$P$ one by one, where each user transmits the next edge that extends the partial path transmitted so far.
  This process is repeated for $N=O_\eps(T)$ times. In~\cite{BR14} it is shown that unless the noise rate exceeds~$1/4$, after $N$~rounds both parties will decode the entire path~$P$. We refer the reader to~\cite{BR14} for a full description of the protocol and correctness proof. We now extend the analysis for the case of channels with errors and erasures.

To simplify the explanation, assume that the players wish to exchange, at each round, a transmission over $\Gamma'=\{0,\ldots,N\} \times \{0,1\}^{\le2}$. Intuitively, the transmission $(e,s)\in\Gamma'$ means ``extend the path~$P$ by taking at most two steps defined by~$s$ starting at the child of the edge I have transmitted at transmission number~$e$''.

Since $\Gamma$ is not of constant size, the  symbol $(e,s)$ is not communicated directly over the channel, but is encoded in the following manner.
Let $\Gamma = \{ <, 0, 1, >, \}$ and encode each $(e,s)$  into a string $< z>\,\in \Gamma^{\le \log N+2}$ where $z$ is the binary representation of $(e,s)$. 
Furthermore, assume that $|\!<z>\!| \le c\log(e)$ for some constant $c$ we can pick later.
Next, each symbol of $\Gamma$ is encoded via a
$|\Gamma|$-ary tree-code 
with distance parameter $1-\eps$ and label alphabet~$\Sigma=O_\eps(|\Gamma|)$.\footnote{On top of the tree-code encoding, we implicitly perform silence encoding of every symbol in $\Sigma$.}
At time~$n$ Alice sends $a_n\in \Sigma$, the last symbol of $\Tenc((e,s)_1,\ldots,(e,s)_n)=a_1a_2\cdots a_n$,
and Bob receives $\tilde{a}_n\in \Gamma \cup \{\bot\}$, possibly with added noise or an erasure mark
(similarly, Bob sends $b_n \in \Sigma$, and Alice receives~$\tilde{b}_n$).
Let $\Tdec(\tilde{a}_1,\ldots, \tilde{a}_n)$ denote the string Bob decodes at time~$n$ (similarly, Alice  decodes $\Tdec(\tilde{b}_1,\ldots, \tilde{b}_n)$).
For every $i>0$, we denote with~$m(i)$ the largest number such that the first $m(i)$ symbols of $\Tdec(\tilde{a}_1, \ldots, \tilde{a}_i)$  equal to $a_1, \ldots, a_{m(i)}$ and the first $m(i)$ symbols of $\Tdec(\tilde{b}_1, \ldots, \tilde{b}_i)$  equal to $b_1, \ldots, b_{m(i)}$.

Let $\N$ be as defined in Definition~\ref{def:N}. 
We begin by showing that if $m(i) < i$ then many corruptions must have happened in the interval $[m(i)+1, i]$.
\begin{lemma}\label{lem:br1}
$\N (m(i)+1,i) \ge (1-\eps)(i-m(i))$.
\end{lemma}
\begin{proof}
Assume that at time $i$ Bob decodes the string $a'_1, \ldots, a'_i$.
By the definition of $m(i)$,  $a'_1,\ldots, a'_{m(i)} = a_1,\ldots, a_{m(i)}$, and assume without loss of generality that $a'_{m(i)+1} \ne a_{m(i)+1}$.  Note that the Hamming distance
between $\Tenc(a_1, \ldots, a_i)$ and $\Tenc(a'_1, \ldots, a'_i)$ must be at least ${(1-\eps)}{(i-m(i))}$. It is immediate that for Bob to make such a decoding error, $\N_A \ge {(1-\eps)}{(i-m(i))}$.
\qed
\end{proof}
Next, we demonstrate that if some party didn't announce the $k$-th edge by round $i+1$, it must be that the $(k-1)$-th edge wasn't correctly decoded early enough to allow  completing the transmission of the $k$-th edge.
\begin{lemma}\label{lem:br2}
Let $t(i)$ be  the earliest time such that both users announced the first $i$~edges of~$P$ within their transmissions.
For $i\ge 0$, $k \ge 1$, if $t(k) > i+1$, then either
$t(k-1) > i-c\log(i-(t(k-1))$, or
there exists $j$ such that $t(k-1)>m(j)$ and $i-c\log(i-t(k-1)) < j \le i$.
\end{lemma}
\begin{proof}{}
[The proof is taken from~\cite{BR14}, as this claim is independent of the definition of $\N$.]
Without loss of generality, assume that the $k$-th edge of~$P$ describes Alice's move.
Suppose that for any $j$ that satisfies 
$i-c\log(i-t(k-1))<j \le i$ both $t(k-1) \le m(j)$ and $t(k-1) \le i-c\log(i-(t(k-1))$.
Then it must be the case that the first $k-1$ edges of~$P$
have already been announced, and correctly decoded by Alice for any $j$ in the last $c\log(i-t(k-1))$ rounds, yet the $k$th~edge has not. However, by the protocol definition, Alice should announce this edge, and this takes her at most~$c\log(i-t(k-1))$ rounds, thus by round $i+1$ she has completed announcing it, in contradiction to our assumption that $t(k) > i+1$.
\qed
\end{proof}

\goodbreak
Finally, we relate the effective noise rate with the progress of the protocol.
\begin{lemma}\label{lem:br3}
For $i\ge -1$, $k\ge 0$, if $t(k) > i+1$, then there exist numbers $\ell_1, \ldots, \ell_k \ge 0$ such that $\sum_{s=1}^k \ell_s \le i+1$ 
and
$\N(1,i)\ge (1-\eps)(i-k+1-\sum_{s=1}^kc\log(\ell_s+2))$.
\end{lemma}
\begin{proof}
We prove by induction. The claim trivially holds for $k=1$ and for~$i\le0$  by choosing $\ell_s=0$. 
Otherwise, by Lemma~\ref{lem:br2} there are two cases. The first case is that $t(k-1) > i-c\log(i-t(k-1))$. Let $i'=t(k-1)-1$ and $k'=k-1$, thus by the induction hypothesis (on $i',k'$), there exist $\ell_1, \ldots, \ell_{k-1} \ge 0$ with $\sum_{s=1}^{k-1}c\log( \ell_s) \le t(k+1)$ such that
\begin{align*}
\N(1,i) \ge \N(1,'i)& \ge (1-\eps)\left( \left(t(k-1)-1\right)-(k-1)+1-\sum_{s=1}^{k-1} c\log(\ell_s +2)\right)  \\
& = (1-\eps)\left( i-k+1-\sum_{s=1}^{k-1} c\log(\ell_s +2) - (i-t(k-1))\right)
\end{align*}
Set $\ell_k= i-t(k+1)$ to complete this case.

In the other case there exists $j$ such that $m(j)< t(k-1)$ and $i-c\log(i-t(k-1)) < j \le i$. In this case we can write
\[
\N(1,i) = \N(1,m(j)) + \N(m(j)+1,i).
\] 
The second term is lower bounded by $\N(m(j)+1,j)$, which by Lemma~\ref{lem:br1} is lower bounded by $(1-\eps)(j-m(j))$. 
We use the induction hypothesis to bound the first term  (with $i'=m(j)-1$ and  $k'=k-1$) to get
\begin{align*}
\N(1,m(j)) \ge \N(1,m(j)-1) &\ge (1-\eps)\left( m(j)-1 -(k-1)+1- \sum_{s=1}^{k-1}c\log(\ell_s+2) \right)
 \\
& =  (1-\eps)\left(j -k+1- \sum_{s=1}^{k-1}c\log(\ell_s+2) -j +m(j) \right)
\end{align*}
for $\ell_1, \ldots, \ell_{k-1} \ge 0$ such that $\sum_{s=1}^{k-1}\ell_s < m(j)$.
Take $\ell_k=i-t(k-1)$. Since $t(k-1)\ge m(j)$ we get that $\sum_{s=1}^{k}\ell_s < m_j + (i-m(j)) < i +1$ and 
\begin{align*}
\N(1,i) &\ge \N(1,m(j)) + \N(m(j)+1,i)  \\
& \ge (1-\eps) \left( j-k+1 - \sum_{s=1}^{k-1}c\log(\ell_s+2)  \right) \\
& \ge  (1-\eps) \left( i-k+1 - \sum_{s=1}^{k-1}c\log(\ell_s+2) -i+j \right)
\end{align*}
Which completes the proof since for this case
$i-j < c\log(i-t(k-1))=c\log(\ell_k)$.
\qed
\end{proof}
We can now complete the proof of Lemma~\ref{lem:Nhigh}. 
Suppose the protocol $\pi_{BR}$ has failed, thus $m(N) <t(T)$. 
By Lemma~\ref{lem:br3} we have $\ell_1,\ldots,\ell_T \ge 0$ that satisfy $\sum_{s=1}^T \ell_s < m(N) < N$ and
\begin{align*}
\N(1,N) &\ge \N(1,m(N)-1) + \N(m(N)+1,N) \\
& \ge (1-\eps) (m(N)-T-\sum_{s=1}^T c\log(\ell_s+2)) + (1-\eps)(N-m(N)) \\
& \ge (1-\eps)\left(N - T - cT\log\left(\frac1T\sum_{s=1}^T (\ell_s+2)\right)\right) \\
&\ge (1-\eps)\left(N - T - cT\log\left(\frac{m(N)}T+2\right)\right),
\end{align*}
where the second transition is due to the concavity of the $\log$ function.
Setting, for instance, $N=T\frac{c^2}{\eps}\log(\eps^{-1})$ gives 
\begin{align*}
\N(1,N) & \ge (1-\eps)\left(N - \eps N / c^2\log(\eps^{-1}) - \eps N \log(3N/T)/c\log(\eps^{-1})\right) \\
&=   (1-\eps)\left (N - \eps N / c^2\log(\eps^{-1})  - \eps\log(3c^2\eps^{-1}\log(\eps^{-1}))/c\log(\eps^{-1}) \right) \\
&=  (1-\eps)\left (1-  \eps  \frac{1+\log(\eps^{-1}) + c\log(3c^2\log(\eps^{-1}))}{c^2\log(\eps^{-1})}
\right) N \\
&> (1-\eps)^2N ,
\end{align*}
for a large enough constant~$c$.
\qed
\end{proof} %

\section{Proof of Theorem~\ref{thm:main-shared}}
\label{app:protocolSharedRand}

In this section we provide the detailed proof for Theorem~\ref{thm:main-shared}. For convenience, we re-state the theorem below.

\smallskip\noindent
\textbf{Theorem~\ref{thm:main-shared}.} 
\textit{%
For any  small enough constants $\eps >0$ and for any function~$f$, 
there exists an interactive protocol in the $\mmain$~model 
with round complexity $O(\comm_f)$ such that,
if the adversarial relative corruption rate is at most~$1-\eps$,
the protocol correctly computes~$f$ with overwhelming success probability over the choice of the shared random string.
}\medskip

We begin with a short motivation for our construction.
Our starting point is the protocol of Theorem~\ref{thm:ProtocolHalf}, i.e., concatenating~\cite{BR14} with silence encoding. We need to deal with two issues: deletion of labels (erasures) and altering labels (errors). 
First we take care of the errors, which is done by the technique of the so called \emph{Blueberry code}~\cite{FGOS15}.
A \bc with parameter $q$ encodes each symbol in~$\Sigma$ into a random symbol in $\Gamma$, where $|\Sigma|/|\Gamma| < q$. Since the mapping $\Sigma\to\Gamma$ is unknown to the adversary, any change to the coded label will be detected with probability $1-q$ and the transmission will be considered as an erasure. By choosing~$q$ to be small enough (as a function of~$\eps$), we can guarantee that the adversary cannot do to much harm by changing symbols. 

Next, we need to deal with the more problematic issue of erasures. The problem is that the \cite{BR14} protocol is \emph{symmetric}, that is, Alice and Bob speak the same amount of symbols. Thus, a successful corruption strategy with relative noise rate~$1/2$ just deletes all Alice's symbols. To overcome this issue we need to ``break'' the symmetry. We will do that by sending indication of deleted labels: If Alice's label was deleted, Bob will tell her so, and she will send more copies of the deleted label. If all of these repeated transmissions are deleted again, Bob will indicate so and Alice will send again more and more copies of that label. This continues until the total amount of re-transmissions surpasses the amount of transmissions in the noiseless scenario.
This breaks the symmetry: if Eve wishes to delete all the copies she will end up causing Alice to speak more, which forces Eve to delete the additional communication as well, which in turn forces her into increasing the average relative noise rate she introduces.  

The remaining issue is to prevent Eve from causing the parties to communicate many symbols
without her making many corruptions, e.g., by forging Bob's feedback to make Alice send unnecessary copies of her labels. This is prevented by the \bc: such an attack succeeds with very small probability that makes it unaffordable.

\begin{proof}\textbf{(Theorem~\ref{thm:main-shared}.)}\ 
The protocol is based on the protocol of Theorem~\ref{thm:ProtocolHalf} (i.e., on the scheme of~\cite{BR14}), yet replacing each label transmission with an adaptive subprotocol that allows retransmissions of deleted symbols. Furthermore, each transmission is encoded via a \bc~\cite{FGOS15}, which allows the parties to notice Eve's attack most of the times.

Fix $\eps>0$ to be some small enough constant. 
For any noiseless protocol $\pi$ of length $T$ we will simulate $\pi$ in the $\mmain$ model 
using the following procedure, defined with parameters $k=\eps^{-1}, t=k\eps^{-1}, q < (kt)^{-2}$.
\begin{enumerate}
\item The parties perform the scheme~\cite{BR14} for $N=O(T/\eps)$ rounds.
\item In each round, each label is transmitted via the following process:
	\begin{enumerate}
	\item The sender encodes the label via a \bc with parameter~$q$, and sends it encoded with a $k$-silence encoding.
	\item Repeat for $t$~times: 
	\setlength{\leftmarginiii}{12pt}
	\begin{itemize}
	\item
	 if the receiver hasn't received a valid label, he sends back a ``repeat-request'' encoded using the \bc and a $1$-silence encoding. [Otherwise, he does nothing.] 
	\item
	each time the sender gets a valid repeat-request, he sends the original message transmission again encoded with a (fresh instance of)  \bc, and $k$-silence encoding. 
	\end{itemize}
	\item The receiver sets that round's output to be the first valid label he received during step (b), or~$\bot$ if all the $t$~repetitions are invalid (either erased, or marked invalid by the \bc).
	\end{enumerate}
\end{enumerate} 

First we note that indeed the protocol takes at most $O(N)=O(T)$ rounds, since the BR protocol takes $O(N)$ rounds, each of which is expended by at most $O(kt|\Gamma|)=O_\eps(1)$ rounds.

We now show that the above protocol achieves noise rate up to $1-O(\eps)$.
We split the protocols into \emph{epochs} , 
where each epoch corresponds to a single
label transmission of the \cite{BR14} simulation
(i.e., the label and its repetitions are a single epoch).

Next, we divide the epochs into two disjoint sets:  \emph{deleted} and  \emph{undeleted} epochs. The former consists of any epoch in which the receiver outputs~$\bot$. The set of undeleted epochs is split again into two disjoint sets: ``lucky'' and ``non-lucky''. A lucky epoch is any epoch in which Eve's corruption is not detected by the \bc.

For each epoch $e$ we define $c(e)$ as the communication made by the parties in this epoch, and $r(e)$ as the rate of noise made by Eve in this epoch, that is, the number of corruptions Eve makes in the epoch is~$r(e)c(e)$.
The global noise rate is a \emph{weighted} average of the noise rate per epoch, where each epoch is weighted by the communication in that epoch.

Fix a noise pattern~$E$ for the adversary, and assume that the simulation process fails with noise pattern~$E$. 
Lemma~\ref{lem:Nhigh} tells us that the number of deleted epochs plus twice the number of incorrect epochs (where the receiver outputs a wrong label) must be at least $(1-\eps)^2N$. 
Note that incorrect epoch must be lucky, and the probability for an epoch to be lucky is at most $q\cdot 2t$, since there are at most $2t$ messages in each epoch and a probability at most~$q$ to break the \bc for a single message.
Hence, for our choice of~$q$,  the number of \emph{deleted} epochs is at least~$(1-3\eps)N$  with overwhelming probability.

First, we analyze deleted epochs. The following is immediate.
\begin{claim}
In each deleted epoch $e$, the noise rate is 
\[
r(e) \ge \frac{k(t+1)}{(k+1)t+k} =  1-\frac{t}{kt+t +k} \ge 1-\eps
\]
while the communication is $c(e) \ge k+t > \eps^{-2}+\eps^{-1}$.
\end{claim}

Next we analyze the non-deleted epochs. We note that~$c(e)$ in this case ranges between $\eps^{-1}$ to~$\eps^{-3}$, and we now relate $r(e)$ to $c(e)$. It is crucial that whenever~$c(e)$ exceeds~$\eps^{-2}$, the amount of noise will be high enough, to maintain a global noise rate of almost one. 

First, we deal with ``lucky'' epochs. 
For simplification, we assume that if $e$ is ``lucky'', then the noise rate is 0, and the communication is the maximal possible~$kt+t+1$.
We will choose $q$ to satisfy~$2qt\cdot (kt+t+1) \ll 1$, so that the effective noise added by such epochs is negligible.

\goodbreak
Then, we need to relate $c(e)$ and $r(e)$ for the rest of the epochs.
\begin{claim}
If $e$ is a non-lucky non-deleted epoch, then
\[
r(e)c(e) \ge \max \big ( 0\, ,\,  c(e) - 2k -1 \big ).
\]
\end{claim}
\begin{proof}
If the epoch is not lucky, then it must have concluded correctly (Bob has eventually received the correct label).
Thus the only attack Eve can perform in order to increase $c(e)$, is to delete Bob's reply-request and Alice's answers up to some point (in addition to corrupting Alice original label).

Thus, Eve must block the first $k$ symbols sent by Alice (as long as $c(e)>k$), but she must not block the last $k$ symbols made by Alice. 
Assume Eve only blocks Bob's reply-requests. Then she can block $x \le t-1$ such requests and make $x$ corruptions out of total communication $x+2k+1$.
Another possible attack is to let $y$ of Bob's repeat-request go through (as long $y+x \le t-1$) but delete Alice's replies (again, except for the last one). This will cause $yk+x$ corruptions out of communication $(y+2)k+x+1$. 
\qed
\end{proof}

The relative noise rate caused by any noise pattern $E$ that results in a failed instance of the simulation is thus bounded by
\begin{align*}
\erate^{\text{adp}}(E) &\ge 
\frac{\sum_{e:\text{ deleted}}r(e)c(e)+\sum_{e: \text{ lucky}}0 + \sum_{e:\text{ correct non-lucky}}\max(0 , c(e)-2k-1)}{\sum_e c(e) } \\
&\ge \frac{(1-\eps)\sum_{e: \text{ deleted}}c(e) + \sum_{e:\text{ correct non-lucky}}\max(0 , c(e)-2k-1)}{ \sum_{e: \text{ deleted}}c(e) + \sum_{e: \text{ lucky}}(kt+t+1)+\sum_{e:\text{ correct non-lucky}}c(e)}.
\end{align*}
Recall that with very high probability $>(1-3\eps) N$ epochs are deleted 
and at most $2qtN$ epochs are lucky.
Thus $ \sum_{e: \text{ lucky}}(kt+t+1)$ is upper bounded by $2qt(kt+t+1)N$ with high probability. We take $q \ll 1/2t(kt+t+1)$ and neglect this term in the denominator.

Now split the correct non-lucky epochs to two sets: $B_0=\{ e \text{} \mid c(e) \le \eps^{-1.5}\}$ contains epochs with ``low'' communication and~$B_1$ contains all the other correct non-lucky epochs (with ``high'' communication).
For small enough~$\eps$,
\begin{align*}
\erate^{\text{adp}}(E)&\gtrapprox
\frac{(1-\eps)\sum_{e: \text{ deleted}}c(e) + \sum_{B_1}(c(e)-2\eps^{-1})}
{\sum_{e: \text{ deleted}}c(e) +|B_0|\eps^{-1.5} + \sum_{B_1}c(e)} \\
& \ge  \frac{(1-\eps)\sum_{e: \text{ deleted}}c(e) - |B_1|2\eps^{-1}}
{\sum_{e: \text{ deleted}}c(e) +|B_0|\eps^{-1.5} } \\
&\ge 1-O(\eps),
\end{align*}
since $\sum_{e: \text{ deleted}}c(e) \ge (1-3\eps)(\eps^{-2}+\eps^{-1})N$ and  $|B_0|+|B_1| \le (1+3\eps)N$, with very high probability.
\qed
\end{proof}

\end{document}